\documentclass[11pt]{article}  
\usepackage{bm}
\usepackage{amscd}
\usepackage{amsmath}
\usepackage{amssymb}
\usepackage{amsthm}
\usepackage{graphicx}
\usepackage{mathtools}
\usepackage{enumerate}
\usepackage{color}
\usepackage{cleveref} 

\newtheorem{definition}{Definition}
\newtheorem{theorem}{Theorem}
\newtheorem{lemma}[theorem]{Lemma}

\newcommand{\ve}{\varepsilon}

\newcommand{\beq}[2]{\begin{equation}\label{#1}#2\end{equation}}

\newcommand\given[1][]{\:#1\vert\:}

\DeclarePairedDelimiter{\defaultDelim}{[}{]}

\DeclareMathOperator{\capPr}{\sf Pr}
\renewcommand{\Pr}[2][]{\capPr_{#1}\defaultDelim*{#2}}
\DeclareMathOperator{\capE}{\sf E}
\newcommand{\E}[2][]{\capE_{#1}\defaultDelim*{#2}}

\newcommand{\set}[1]{\left\{#1\right\}}

\DeclareMathOperator{\capExp}{\sf exp}
\renewcommand{\exp}[2]{\capExp_{#1}\left({#2}\right)}
\renewcommand{\exp}[2][]{\capExp_{#1}\brac{#2}}

\def\cA{{\mathcal A}}

\def\a{\alpha}

\def\e{\varepsilon}

\def\th{\theta}

\newcommand{\brac}[1]{\left(#1\right)}

\newcommand{\bfrac}[2]{\left(\frac{#1}{#2}\right)}
\crefformat{equation}{(#2#1#3)}

\def\cE{{\mathcal E}}

\begin{document}
\title{Balanced Allocation Through Random Walk}
\author{Alan Frieze\thanks{Department of Mathematical Sciences, Carnegie Mellon University, Pittsburgh PA15213. Research supported in part by NSF grant DMS0753472}\and Samantha Petti\thanks{School of Mathematics, Georgia Tech., Atlanta, GA30313. This material is based upon work supported by the National Science Foundation Graduate Research Fellowship under Grant No. DGE-1650044.}} 
\maketitle

\begin{abstract}
We consider the allocation problem in which $m \leq (1-\ve) dn $ items are to be allocated to $n$ bins with capacity $d$. The items $x_1,x_2,\ldots,x_m$ arrive sequentially and when item $x_i$ arrives it is given two possible bin locations $p_i=h_1(x_i),q_i=h_2(x_i)$ via hash functions $h_1,h_2$. We consider a random walk procedure for inserting items and show that the expected time insertion time is constant provided $\ve = \Omega\left(\sqrt{ \frac{ \log d}{d}} \right).$
\end{abstract}

\section{Introduction}
We consider the following allocation problems. We have $m$ items that are to be allocated to $n$ bins, where each bin has space for $d$ items. The items $x_1,x_2,\ldots,x_m$ arrive sequentially and when item $x_i$ arrives it is given two possible bin locations $p_i=h_1(x_i),q_i=h_2(x_i)$ via hash functions $h_1,h_2$. We shall for the purpose of this paper assume that $p_i\neq q_i$ for $i\in [m]$ and that $(p_i,q_i)$ is otherwise chosen uniformly at random from $[n]^2$. This model is explicitly discussed in Dietzfelbinger and Weidling \cite{DW}.

Suppose now that $m\leq d(1-\e)n$ where $m,n$ grow arbitrarily large and that $\e>0$ is small and independent of $n$. They prove the following:
\begin{enumerate}
\item \label{enum1} If $d\geq 1+\frac{\log(1/\e)}{1-\log 2}$ then w.h.p.\footnote{A sequence of events $(\cE_n,n\geq 0)$ is said to occur with high probability (w.h.p.) if $\lim_{n\to\infty}\Pr{\cE_n}=1$.} all the items can placed into bins. 
\item If $d> 90\log(1/\e)$ then the expected time for a Breadth First Search (BFS) procedure to insert an item is at most $(1/\e)^{O(\log d)}$.
\end{enumerate}
This model is related to a $d$-ary version of Cuckoo Hashing (Pagh and Rodler \cite{PR}) that was discussed in Fotakis, Pagh, Sanders and Spirakis \cite{FPSS}. Here there are $d$ hash functions and the bins are of size one. This latter paper also uses BFS to insert items.

Item insertion in both of these models can also be tackled via random walk. For $d$-ary Cuckoo Hashing, Frieze, Mitzenmacher and Melsted \cite{FMM} and Fountoulakis, Panagiotou and Steger \cite{FPS} gave $O((\log n)^{O(1)})$ time bounds for random walk insertion and more recently Frieze and Johansson \cite{FJ} proved an $O(1)$ time bound on random walk insertion, for $d$ sufficently large.

The authors of \cite{DW} ask for an analysis of a random walk procedure for inserting an item. They ask for bounds of $O(\log 1/\e)$ insertion time while maintaining $d=O(\log 1/\e)$. While we cannot satisfy these demanding criteria, in this note we are able to establish constant expected time bounds with a larger value of $d$. We first describe the insertion algorithm. We say a bin is saturated if it contains $d$ items. 

\medskip

\noindent{
{\bf Random Walk Insertion: RWI}
}
\begin{tabbing}
xxx\= xxx\= xxx\= \kill
{\bf For} $i=1$ to $m$ {\bf do}\\
{\bf begin}\+\\
Generate $p_i,q_i$ randomly from $[n]$\\
{\bf if} either of bins $p_i,q_i$ are not saturated, then assign item $i$ arbitrarily to one of them.\\
{\bf if} both bins $p_i,q_i$ are saturated then {\bf do}\+\\
{\bf begin}\\
Choose $b$ randomly from $\set{p_i,q_i}$.\\
{\bf repeat}\hspace{-70pt}{\bf A}\+\\
Let $x$ be a randomly chosen item from bin $b$.\\
Remove $x$ from bin $b$ and replace it with item $i$.\\
Let $c$ be the other bin choice of item $x$.\\
$b\gets c$.\-\\
{\bf until} bin $b$ is unsaturated.\-\\
Place item $x$ in bin $b$.\\
{\bf end}\-\\
{\bf end}
\end{tabbing}
Let $r_i$ denote the number of steps in loop A of algorithm RWI. Then,
\begin{theorem} \label{main} Let $m \leq (1-\ve) dn$. Then for some absolute constant $M>0$,
\beq{th1}{
\E{r_i}\leq \frac{4M}{\e^2}\text{ w.h.p. for }i\in [m]\text{ provided }\ve \geq \sqrt{ \frac{ M \left( \log( 4d) +1\right)}{d}}.
}
\end{theorem}
In the analysis below, we take $M=96$. It goes without saying that we have not tried to optimze $M$ here.

There are two sources of randomness here. The random choice of the hash functions and the random choices by the algorithm. The w.h.p. concerns the assignment of items to bins by the hash functions and the $\E{r_i}$ is then the conditional expectation given these choices.
\section{Graphical Description} \label{technique}
We use a digraph $D'$ to represent the assignment of items to bins. Each bin is represented as a vertex in $D'$ and item $i$ is represented as a directed edge $(p_i,q_i)$ or $(q_i,p_i)$ that is oriented toward its assigned bin. We say a vertex is {\em saturated} if its in-degree in $D'$ is $d$. As the algorithm is executed, we in fact build two digraphs $D$ and $D'$ simultaneously. 

We now describe the insertion of an item in terms of $D,D'$. Let $x$ and $y$ denote the two randomly selected bins for an item. We place a randomly oriented edge between $x$ and $y$ in $D$. If $x$ and $y$ are unsaturated in $D'$, then we place the edge in the same orientation as in $D$. If $x$ or $y$ is saturated in $D'$, then we place the edge in $D'$ according to the algorithm RWI, which may require flipping edges in $D'$. Repeat for all items. Note that $D$ is a random directed graph with $(1-\ve)dn$ edges. The undirected degree of each vertex in $D$ is the same as in $D'$. However, the directed degrees will vary.  Let $D_t$ and $D'_t$ denote the respective graphs after $t$ edges have been inserted. 

We compute the expected insertion time after $(1-\ve)dn$ items have been added by analyzing $D'$. The expected time to add the next item is equal to the expected length of the following random walk in $D'$. Select two vertices $x$ and $y$. If either is unsaturated no walk is taken, so we say the walk has length zero. Otherwise, pick a vertex at random from $\{x, y\}$ and walk ``backwards" along edges oriented into the current vertex until an unsaturated vertex is reached. We call this the ``replacement walk."

Let $G$ denote the common underlying graph of $D,D'$ obtained by ignoring orientation. In order to compute the expected length of the replacement walk, we analyze the the structure of the graph $G_S$ induced by a set $S$ which contains all saturated vertices in $G$. In \Cref{sat}, we show that the expected number of components of size $k$ among saturated vertices decays geometrically with $k$ and that each component is a tree or contains precisely one cycle. 
 In \Cref{time} we show that since the components of $G_S$ are sparse and the number of components decays geometrically with size, the expected length of a replacement walk is constant. 

\section{Saturated Vertices}\label{sat}
In this section we describe the structure induced by $G$ on the set of saturated vertices.
First we define a set $S$ that is a superset of all saturated vertices. 
\begin{definition} \label{s} Let $A$ be the set of vertices of $D$ with in-degree at least $d-1$ in $D$ and $T_0= \emptyset$. Given $ A, T_0, \dots T_k$, let $T_{k+1}$ be all the vertices of $V \setminus \left( A \cup T_0 \cup T_1 \cup \dots \cup T_k \right)$  with at least two neighbors in  $A \cup T_0 \cup T_1 \cup \dots \cup T_k$. Let  $T= \bigcup T_i$ and $S= A \cup T$.  \end{definition}

\begin{lemma} \label{saturated}
The set $S$ defined above contains all saturated vertices. 
\end{lemma}

\begin{proof}  We prove the statement by induction on $S_t$ the set of saturated vertices after the $t^{th}$ edge is added. Since $S_0 = \emptyset$, $S_0\subseteq S$ vacuously. Assume $S_t \subseteq S$. Note that the addition of a single edge can cause at most one vertex to become saturated. If $S_t = S_{t+1}$, $S_{t+1}\subseteq S$ trivially. For the other case, let $v$ be the vertex that became saturated as a result of the addition of the $(t + 1)^{st}$ edge. If $v$ has in-degree at least $ d-1$ in $D$ then $v\in  A \subseteq S$. Otherwise, there must exist two edges $\{u, v\}$ and $\{w, v\}$ that are oriented out of $v$ in $D$ and are oriented into $v$ in $D'$ at time $t + 1$. Since the orientation of an edge differs in $D$ and $D'$ only if one of its ends is saturated, it follows that $u$ and $w$ must be saturated at time $t$. By the inductive hypothesis, $u, w \in S$. Therefore, since $v$ is adjacent to two vertices of $S$, $v \in S$. It follows $S_{t+1} = S_t \cup \{v\} \subseteq S$.
\end{proof}

Next we analyze the structure that $G$ induces on $S$. Let $G_S$ be the subgraph of $G$ induced by $S$. The following lemma states that at least half of the vertices of each component of $G_S$ are in $A$. This fact is crucial for proving \Cref{num comps}, which gives an upper bound on the number of components of $G_S$  of size $k$ that decays geometrically with $k$. 

\begin{lemma}
Let $S= A\cup T$ as in \Cref{s} and let $G_S$ be the graph on $S$ induced by $G$. Each spanning tree of a component in $G_S$ has the following form:
\begin{enumerate}
\item A set $K$ of $k$ vertices.
\item A set $L=K\cap A$ of size $\ell$. We suppose that $L$ induces a forest with $s$ components and $\ell-s$ edges.
\item A set $K\setminus L=K \cap T=\set{v_1,v_2,\ldots,v_{k-\ell}}$ where $v_i$ is saturated by the algorithm after $v_{i-1}$,  $v_i$ has $s_i\geq 2$ neighbors in $L\cup \set{v_1,v_2,\ldots,v_{i-1}}$, and $s_1+s_2+\cdots+s_{k-\ell}=k-1-(\ell-s)$.
\item $\ell\geq \frac{k+2}{2}.$
\end{enumerate}
\end{lemma}

\begin{proof} 
As described in the proof of \Cref{saturated} each vertex in $T$ is adjacent to at least two vertices in $S$. Therefore $(1),(2), (3)$ hold. For $(4)$, note
$$2(k-\ell)\leq s_1+s_2+\cdots+s_{k-\ell}=k-1-(\ell-s)$$
which implies  
$$k-\ell\leq s-1\leq \ell-2,$$
and (4) follows directly. 
\end{proof}

\begin{lemma} \label{num comps}
Let $S$ be as in \Cref{s} and let $M=96$. When 
$$\sqrt{ \frac{ M\left( \log( 4d) +1\right)}{d}} \leq \ve,$$
the expected number of components of $G_S$ of size $k$ is bounded above by 
$$ \frac{n}{k^2} \exp{-\frac{\ve^2dk}{M}}.$$ 
\end{lemma} 

\begin{proof}
We use an upper bound on the expected number of spanning trees of size $k$ in $G_S$ as an upper bound on the number of components of $G_S$ of size $k$. 

Let $T$ be a tree on a specified labeled set of $k$ vertices, and let $L$ be a specified subset of vertices of $V(T)$ of size $\ell$. Let $\cA_{T,L}$ be the event that $T$ is present in $G$ and all vertices in $L$ are in $S$.  Let $I(L)$ be the sum of the in-degrees of the vertices in $L$, and let $I^*(L)$ be $I(L)$ minus the number of edges in $T$ oriented into a vertex in $L$. We compute
\begin{align} 
\Pr{ \cA_{T,L} } &\leq \Pr{ T \text{ appears in } G} \Pr{  I(L) \geq \ell (d-1)\given T \text{ appears in } G }\nonumber \\
&\leq \Pr{ T \text{ appears in } G} \Pr{ I^*(L) \geq \ell (d-1)-(k-1) \given T \text{ appears in } G }\label{middle}\\
&\leq \bfrac{2d}{n}^{k-1} \exp{-\frac{\ve^2dk}{96}} \label{final}
\end{align}
{\bf Explanation of \Cref{middle}.} If the tree $T$ is present in $G$, then $I(L)- I^\ast(L)$ is at most $k-1$. Therefore $I^\ast(L)$ must be at least $\ell (d-1)-(k-1)$.

\noindent {\bf Explanation of \Cref{final}.} The term $\bfrac{2d}{n}^{k-1}$ can be explained as follows: Let $m=(1-\e)dn$ be the number of pairs offered to algorithm RWI. Each pair has probability $1/\binom{n}{2}$ of being a particular edge. The probability that $G$ contains $k-1$ given edges is then at most 
$$m^{k-1}\bfrac{2}{n(n-1)}^{k-1}=\bfrac{2(1-\e)d}{n-1}^{k-1}\leq \bfrac{2d}{n}^{k-1}.$$ 
The term $\exp{-\frac{\ve^2dk}{96}}$ can be explained as follows: $I^\ast(L)$ is dominated by the sum of $\ell(n-\ell)$ independent Bernouilli random variables, $\xi_{u,v}$, each of which corresponds to an ordered pair $(u, v)$ where $u\notin L,v \in L$ and $\set{u,v} \not \in T$.  There are at most  $\ell (n-1) -(k-1)$ such pairs.  Here $\Pr{\xi_{u,v}=1}\leq \frac{m}{2\binom{n}{2}}=\frac{(1-\ve)d}{n-1}$ bounds the probability that the edge $\set{u,v}$ exists in $G$ and oriented from $u$ to $v$ in $D$. The existence of edges are negatively correlated since the number of edges is bounded above.Therefore, $I^\ast(L)$ is dominated by $X_\e \sim Bin \left(\ell( n -1) , \frac{(1-\ve)d}{n-1} \right)$. 

Thus if
$$p^*=\Pr{ I^*(L) \geq \ell (d-1)-(k-1)\given T \text{ appears in } G }$$
then
\begin{align*}
p^* &\leq \Pr{ X_\e \geq \ell (d-1)-(k-1)}\nonumber\\
&\leq \Pr{ X_\e  \geq \ell\left(d  -3\right) }\nonumber\\
&\leq \Pr{ X_\e \geq \left( 1+ \theta\right) \E{X_\e}},
\end{align*}
where $\theta=\frac{\e d-3}{(1-\e)d}$. Note that $\th\leq1$ iff $\e\leq \e_0=\frac{d+3}{2d}$.

The Chernoff bounds then imply that if $\e\leq 1/2\leq \e_0$ then
$$p^*\leq \exp{-\frac13\cdot\bfrac{\e d-3}{(1-\e)d}^2\ell d\brac{1-\e}}\leq e^{-\e^2\ell d/12}\leq e^{-\e^2k d/24}$$
since $\e d\geq 6$. When $\e>1/2$ we see that $X_\e$ is dominated by $X_{1/2}$ and then replacing $\e^2$ by $(\e/2)^2$, we have \Cref{final}.

We now compute the expected number of trees of size $k$. Let $Z_k$ denote the number of spanning trees in $G_S$ with $k$ vertices. Then
\begin{align*}
\E{Z_k}&\leq \binom{n}{k}k^{k-2}\sum_{\ell\geq k/2}\binom{k}{\ell}\Pr{\cA_{T,L}}\\
&\leq \frac{(ne)^k}{k^2} 2^{k-1} \bfrac{2d}{n}^{k-1} \exp{-\frac{\ve^2dk}{96}}\\
&=\frac{n}{k^2} \exp{ k \left( \log(4d) +1- \frac{\ve^2 d}{96} \right)}\\
&\leq \frac{n}{k^2} \exp{-\frac{\ve^2dk}{M}}.
\end{align*}
\end{proof}

Finally, we give a high probability bound on the size of the maximum component in $G_S$ and show that with high probability each component of $G_S$ contains at most one cycle. 

\begin{lemma} \label{typical} 
Let $ \ve \geq \sqrt{ \frac{ M  \left( \log( 4d) +1\right)}{d}}$ and $S$ be as in \Cref{s}. Then, w.h.p., the maximum size of a component of $G_S$ is at most
$$k_0=\a\log n\text{ where } \a=\frac{M}{\e^2d},$$ 
and each component contains at most one cycle. 
\end{lemma}

\begin{proof}
Let $Z_k$ be the expected number of components of $G_S$ of size $k$. We use  \Cref{num comps} to compute 
$$\Pr{  \exists \text{ a component of size $>k_0$ }} \leq \sum_{k_0<k<n} \E{ Z_k} \leq n\sum_{k_0<k<n} k^{-2}\exp{-\frac{\ve^2dk}{M}}= o(1).$$ 
It follows that w.h.p. the maximum size of a component of $G_S$ is at most $k_0$.

Next we show that  w.h.p. each component of $G_S$ contains at most one cycle. Let $Y_k$ be the number of sets $Q$ of size $k$ with $e(Q) \geq |Q|+1$.   We bound the number structures on $k$ vertices consisting of a tree plus two edges by $k^{k+2}$ and compute 
\begin{align*}
\Pr{ \exists\text{$Q$: $|Q| \leq k_0$, $e(Q) \geq |Q|+1$}}&\leq \sum_{k \leq k_0} \E{Y_k}\\
&\leq \sum_{k\leq k_0}\binom{n}{k}k^{k+2}\bfrac{2d}{n}^{k+1}\\
&\leq \frac{2d}{n}\sum_{k\leq k_0}k^2(2de)^k\\
&=o(1).
\end{align*}
\end{proof}

\section{Expected insertion time}\label{time}

Recall from \Cref{technique} the definition of a replacement walk. To bound the the expected length of the replacement walk, we first show that a random walk on $G'$ that starts in a size $k$ component of $G_S$ will remain in the component for at most $2k$ steps in expectation. 

\begin{lemma} \label{walk bound}
Let  $G$ satisfy the conditions of \Cref{typical}, and let $S$ be as in \Cref{s}. 
Consider a random walk on $G'$ begining in a size $k$ component of $G_S$. Then expected number of steps this random walk takes before leaving $G_S$ is bounded above by $2 k $.  
\end{lemma}

\begin{proof} 
Consider a random walk on $G'$ beginning in a size $k$ component $C$  of $G_S$. Note the edges  of $G'$ are oriented, so the walk does not back track. Therefore, if $C$ is a tree, a random walk can take at most $k$ steps in $C$ before leaving $C$. 

Next suppose $C$ has a cycle of length $\ell$. Since $C$ has only one cycle,  in any walk in $C$ all edges of the cycle must be visited consecutively (i.e. it is not possible to take a walk on the cycle, leave the cycle, and then return to the cycle). Let $v$ be the first vertex visited by the walk that is on the cycle. The expected number of times the walk completes the entire cycle at this point is
 $$ \sum_{i=1}^ \infty \bfrac{1}{d^\ell}^i \left( 1- \frac{1}{d^\ell}\right) i= \frac{1}{d^\ell-1}.$$
Since the walk cannot return to the cycle after the walk leaves the cycle, the expected length of the walk before it leaves $C$ is at most 
$$k + \frac{\ell}{d^\ell-1} \leq 2k. $$
\end{proof}

Finally, we prove \Cref{main}. 

\begin{proof} (of \Cref{main}).
The expected insertion time is the probability two randomly selected vertices $x$ and $y$ of $G'$ are saturated times the expected length of a replacement walk in the graph $G'$ starting a random saturated vertex. Since the replacement walk ends once an unsaturated vertex is reached, \Cref{walk bound} implies that the expected length of a replacement walk starting at a vertex in  a component of size $k$ is bounded above by $2k$. We compute
\begin{align*}
\E{ \text{ insertion time }} &\leq  2\sum_{i=1}^{k_0} \Pr{ x \in S} \Pr{y \text{ in component of size $k$ in $G_S$ }}k \\
&\leq 2\sum_{i=1}^{k_0} \frac{|S|}{n} \frac{ k \frac{n}{k^2} \exp{\frac{-\ve^2 d k}{M}}}{n} k \\ &\leq 2\sum_{i=1}^{\infty} \exp{-\frac{\ve^2 d k}{M}}\\
&=\frac{2}{1-e^{-\e^2/M}}\\
&\leq \frac{4M}{\e^2}
\end{align*}
\end{proof}
\section{Conclusion}
We have proved an $O(1)$ bound on the expected insertion time of a natural random walk insertion algorithm. The next step in the analysis of this algorithm is to reduce the dependence of $d$ on $\e$.


\begin{thebibliography}{99}
\bibitem{DW} M. Dietzfelbinger and C. Weidling, Balanced Allocation and Dictionaries with Tightly Packed Constant Sized Bins, {\em Theoretical Computer Science} 380 (2007) 47-68. 
\bibitem{FPSS}  D. Fotakis, R. Pagh, P. Sanders, and P. Spirakis, Space Efficient Hash Tables With Worst Case Constant Access Time, {\em Theory of Computing Systems} 8 (2005) 229-248.
\bibitem{FPS} N. Fountoulakis, K. Panagiotou and A. Steger, On the Insertion Time of Cuckoo Hashing, {\em SIAM Journal on Computing} 42 (2013) 2156-2181.
\bibitem{FJ} A.M. Frieze and T. Johansson, On the insertion time of random walk cuckoo hashing, Proceedings of SODA 2017.
\bibitem{FMM} A.M. Frieze, P. Melsted and M. Mitzenmacher, An Analysis of Random-Walk Cuckoo Hashing, {\em SIAM Journal on Computing} 40 (2011) 291-308.
\bibitem{PR} R. Pagh and F. Rodler. Cuckoo Hashing, {\em Journal of Algorithms} 51 (2004) 122-144.
\end{thebibliography}
\end{document}